\documentclass[10pt]{amsart}
\usepackage{amssymb}
\usepackage{mathrsfs}
\usepackage{setspace}
\usepackage{amsmath, bm}
\usepackage{mathrsfs}
\usepackage{xcolor}
\newcommand{\cC}{\mathcal{C}}
\newcommand{\cP}{\mathcal{P}}

\newcommand{\cH}{\mathcal{H}}

\newcommand{\cM}{\mathcal{M}}

\newcommand{\mZ}{\mathbb{Z}}
\newcommand{\mR}{\mathbb{R}}
\newcommand{\mC}{\mathbb{C}}
\newcommand{\f}{\mathfrak{f}}
\newcommand{\fd}{\mathfrak{f}^\dagger}
\newcommand{\fdd}{\mathfrak{f}^{\dagger2}}
\newcommand{\upx}{\partial_{\underline{x}}}
\newcommand{\ux}{\underline{x}}

\newcommand{\gpde}{\upx+ \zeta}

\newcommand{\ul}{\underline}

\newtheorem{theorem}{Theorem}[section]
\newtheorem{lemma}[theorem]{Lemma}
\newtheorem{proposition}[theorem]{Proposition}
\newtheorem{corollary}[theorem]{Corollary}

 \theoremstyle{definition}
\newtheorem{definition}[theorem]{Definition}

\theoremstyle{remark}
\newtheorem{remark}[theorem]{Remark}

\numberwithin{equation}{section}

\begin{document}

\title[Solutions for the L\'evy-Leblond equation]{Solutions for the L\'evy-Leblond or parabolic Dirac equation and its generalizations}

\author{Sijia Bao}
\address{Department of Mathematics\\ Harbin Institute of Technology\\ Harbin, 150001\\ China.}
\email{15b912013@hit.edu.cn}

\author{Denis Constales}
\address{Department of Electronics and Information Systems \\Faculty of Engineering and Architecture\\Ghent University\\Krijgslaan 281, 9000 Gent\\ Belgium.}
\email{denis.constales@ugent.be}

\author{Hendrik De Bie}
\address{Department of Electronics and Information Systems\\Faculty of Engineering and Architecture\\Ghent University\\Krijgslaan 281, 9000 Gent\\ Belgium.}
\email{hendrik.debie@ugent.be}

\author{Teppo Mertens}
\address{Department of Electronics and Information Systems\\Faculty of Engineering and Architecture\\Ghent University\\Krijgslaan 281, 9000 Gent\\ Belgium.}
\email{teppo.mertens@ugent.be}

\begin{abstract}
In this paper we determine solutions for the L\'evy-Leblond operator or a parabolic Dirac operator in terms of hypergeometric functions and spherical harmonics. We subsequently generalise our approach to a wider class of Dirac operators depending on 4 parameters.

\end{abstract}
\date{\today}

\maketitle

\section{Introduction}

 The Dirac operator was originally introduced by Dirac in \cite{D} to study the quantum mechanical behaviour of the electron. This operator arises by constructing a suitable square root of the Laplace operator. This is achieved using gamma matrices in 4 dimensions, or in arbitrary dimension by using Clifford algebras \cite{Lou}.
 
It took a remarkably long time before L\'evy-Leblond realized in \cite{LL} that it is equally possible to construct the square root of the Schr\"odinger operator and, in a similar vein, of the heat operator in $1+3$ dimensions. The resulting operators are called the L\'evy-Leblond operator and the parabolic Dirac operator respectively. The latter name was introduced in \cite{CKS, CSV} because the operator stated there factors the heat operator, which is parabolic. The authors of \cite{CKS, CSV} seemed unaware of the work of L\'evy-Leblond. However, they do give a construction in general $(m+1)$-dimensional space, whereas L\'evy-Leblond only treats dimension $1+3$.

The literature on the L\'evy-Leblond equation investigates it in various physical contexts. We mention supersymmetry and Chern-Simons theory, see e.g. \cite{Gom, Gau, Hor, Du, Hor2}. The direct inspiration for the present paper stems from the recent results in \cite{AK1, AK2}. There the symmetries of the  L\'evy-Leblond operator are computed and organised in terms of a $\mZ_2 \times \mZ_2$-graded Lie superalgebra. This immediately gives rise to two questions:
\begin{itemize}
\item {\bf Q1} Can we construct special classes of solutions of the L\'evy-Leblond equation in arbitrary dimension?
\item {\bf Q2}  How do the symmetries of \cite{AK1, AK2} act on the solutions of the L\'evy-Leblond equation?
\end{itemize}
The present paper aims to tackle {\bf Q1}, not only for the L\'evy-Leblond or parabolic Dirac equation but also for generalizations thereof. We postpone {\bf Q2} to future work.

The construction of special classes of solutions for Dirac operators has recently received a lot of interest. In \cite{SJ} series expansions of such solutions have been constructed. The case of the standard Dirac operator and polynomial expressions in it was treated in \cite{X, SX, R, G}, leading to solutions expressed as series expansions in Bessel functions multiplied with spherical monogenics \cite{Rood, Groen}.  
The hyperbolic Dirac operator was treated in \cite{E}, now using hypergeometric functions.
Several other generalisations of the Dirac operator were treated in \cite{Cac0, Cac1, Cac2}, where again hypergeometric functions arise.
The series expansions mentioned here often serve as a starting point in the construction of generalised Fourier transforms interacting with Dirac operators. We refer the reader to e.g. \cite{DX, D, DOV}.

It is our goal to show that the solutions for the L\'evy-Leblond operator can also be written using suitable hypergeometric functions. This will be achieved as follows. We expand the solutions first as a series of spherical harmonics or spherical monogenics multiplied with radial functions. Using a result from \cite{CSV}, we find the relation between the coefficients and express them as hypergeometric functions. A clever rewriting yields a concise result, see Theorem \ref{thmPara}. Finally, we will extend these techniques to find solutions of the generalised parabolic Dirac operator leading to Theorems \ref{thmGenPara} and \ref{thmGenPara zeta invertible}.

The paper is organised as follows. Section \ref{Sec2} contains the necessary preliminaries on Clifford algebras and the L\'evy-Leblond or the parabolic Dirac operator.
Section \ref{Sec3} determines solutions of the parabolic Dirac operator, culminating in Theorem \ref{thmPara}. In Section \ref{Sec4} we define the generalised parabolic Dirac operator and determine its solution.
We end with conclusions and an outlook for further research.

\section{Preliminaries}
\label{Sec2}

In this section we introduce all concepts necessary for the paper. We mostly follow the notations from \cite{CSV}.

\subsection{Clifford algebras}

Let us consider the vector space $\mR^{1, m+1}$ with basis $( \epsilon, e_1, e_2, \ldots, e_{m+1})$. We use it to construct the Clifford algebra $\cC \ell_{1,m+1}$ of signature $(1,m+1)$ as the algebra generated by the basis elements $\epsilon, e_1, e_2, \ldots, e_{m+1}$ under the relations
\begin{align*}
\epsilon^2&=+1\\
e_j^2 &= -1 \qquad j \in \{1, \ldots, m+1\}\\
e_j e_k + e_k e_j &=0 \phantom{-}\qquad j \neq k\\
\epsilon e_j + e_j \epsilon &=0 \phantom{-}\qquad j \in \{1, \ldots, m+1\}.
\end{align*}
The main involution is defined on the basis elements as

\[
e_j^{\ast} = -e_j, \qquad \epsilon^{\ast} = -\epsilon
\]
and it extends to the whole Clifford algebra $\cC \ell_{1,m+1}$ by

\[
(ab)^{\ast} = a^{\ast}b^{\ast} \qquad \text{and} \qquad (a+b)^{\ast} = a^{\ast} + b^{\ast}.
\]
We now introduce the nilpotent elements $\f =(e_{m+1}- \epsilon)/2$ and $\fd=-(e_{m+1}+ \epsilon)/2$ which satisfy
\begin{align}
\label{nilps}
\begin{split}
\f^2=\fdd=0\\
\f \fd + \fd \f =1\\
\f e_j + e_j \f = \fd e_j + e_j \fd = 0.
\end{split}
\end{align}
The Clifford algebras $\cC \ell_{1,1}$ generated by $\f, \fd$ and $\cC \ell_{0,m}$ generated by  $e_1, e_2, \ldots, e_{m}$ are clearly subalgebras of $\cC \ell_{1,m+1}$. In fact, it holds that $\cC \ell_{1,m+1} = \cC \ell_{1,1} \otimes \cC \ell_{0,m}$.
The elements $\f$ and $\fd$ are crucial to factor the heat operator.
For more details on Clifford algebras we refer the reader to e.g. \cite{Lou}.

\subsection{The L\'evy-Leblond or parabolic Dirac operator}

The standard (orthogonal) Dirac operator is given by 
\[
\upx= \sum_{i=1}^{m}e_{i}\partial_{x_{i}}. 
\]
Its square satisfies $\upx^2= - \Delta$, where $\Delta = \sum\limits_{i=1}^{m}\partial_{x_{i}}^2$ is the Laplace operator on $\mR^m$. The symbol of the Dirac operator $\upx$ is denoted by the vector variable
\[
\ux= \sum_{i=1}^{m}e_{i}x_i
\]
and satisfies $\ux^2 = - |\ux|^2=- \sum_{i=1}^{m}x_i^2$.

Using the nilpotent elements $\f, \fd$ we can now introduce the L\'evy-Leblond or parabolic Dirac operator.
\begin{definition}
\label{def-com}
We define the parabolic Dirac operator as
$$ D_{\underline{x},t}=\upx+\f\partial_{t}+\fd.$$
\end{definition}
It is easy to show, using (\ref{nilps}), that the parabolic Dirac operator indeed factors the heat operator, i.e.
$$(\upx +\f\partial_{t}+\fd )^{2}=-\Delta+\partial_{t}.$$

\begin{remark}
The parabolic Dirac operator can be altered to the L\'evy-Leblond operator by putting imaginary units at suitable positions. The L\'evy-Leblond operator will factor the Schr\"odinger operator instead of the heat operator, see \cite{LL, AK1, AK2}. This choice has no impact on the sequel, apart from the normalizations of constants.

\end{remark}


Using the decomposition $\mathcal{C}\ell_{1,m+1}=\mathcal{C}\ell_{1,1}\otimes\mathcal{C}\ell_{0,m}$ in terms of the nilpotent elements $\f, \fd$, any $\mathcal{C}\ell_{1,m+1}-$valued function $F=F(x,t)$ can be uniquely decomposed as
\begin{align}\label{2.1}
F(x,t)=F^{[0]}(x,t)+\f F^{[1]}(x,t)+\fd F^{[2]}(x,t)+\f\fd F^{[3]}(x,t)
\end{align}
where the components $F^{[i]}, i=0,1,2,3$, are $\mathcal{C}\ell_{0,m}$-valued polynomials in the space-time domain $\Omega=\{(x,t)\in\mathbb{R}^{m}\times\mathbb{R}^{+}\}$. We will say that a $\mathcal{C}\ell_{1,m+1}$-valued function $F(x,t)$ lies in $\mathcal{C}^{[k,l]}(\Omega)$ if and only if $F$ is of class $\mathcal{C}^{k}$ with respect to $x$ and of class $\mathcal{C}^{\ell}$ with respect to $t$. 

The following theorem was proven in \cite{CSV}.
\begin{theorem}
\label{conditionsPDE}
Let $F=F^{[0]}+\f F^{[1]}+\fd F^{[2]}+\f \fd F^{[3]}\in\mathcal{C}^{[2,1]}(\Omega)$. We have $D_{\underline{x},t}F=0$ in $\Omega$ if and only if its components $F^{[i]}, i=0,1,2,3$, satisfy
\begin{equation}
\left\{
            \begin{array}{lr}
            (\Delta-\partial_{t})F^{[i]}=0,\quad i=0,2  \\
             F^{[1]}=-\upx F^{[0]}\\
             F^{[3]}=\upx F^{[2]}-F^{[0]}.
             \end{array}
\right.\notag
\end{equation}
\end{theorem}


Finally, we introduce a few important spaces of polynomials with values in suitable Clifford algebras. Let $\cP_k(\mR^m)$ be the space of $k$-homogeneous polynomials on $\mR^m$. Then we introduce the space of spherical monogenics of degree $k$, see \cite{Rood, Groen}, as
\[
\cM_k(\mR^m, \cC \ell_{0,m})= \left( \cP_k(\mR^m) \otimes \cC \ell_{0,m} \right) \cap \ker{\upx}.
\]
This is a refinement of the space of Clifford algebra-valued spherical harmonics of degree $k$ defined as
\[
\cH_k(\mR^m, \cC \ell_{0,m})= \left( \cP_k(\mR^m) \otimes \cC \ell_{0,m} \right) \cap \ker{\Delta}.
\]
Depending on the situation, we may let the spherical harmonics or monogenics take values in a larger Clifford algebra such as $\cC \ell_{1,m+1}$.


We will need the following technical lemma, which was proven in \cite[Prop 2, p217]{Groen}
\begin{lemma} 
\label{diraclemma}
For $M_{k}(x)\in\cM_k(\mR^m, \cC \ell_{0,m})$ and $g(\rho)$ a scalar function of $\rho=|\underline{x}|$, we have
\begin{align}
\label{diracg}
\upx ( M_{k}(x)g(\rho))&=\underline{x}M_{k}(x)\frac{g'(\rho)}{\rho}\\
\label{diracgx}
\upx (\underline{x}M_{k}(x)g(\rho) )& =-\left((2k+m)M_{k}(x)g(\rho)+ M_{k}(x)\rho g'(\rho)\right).
\end{align}
\end{lemma}

We then immediately obtain the following corollary.
\begin{corollary}
\label{cortechn}
For any integer $\ell\in\mathbb{N}=\{0,1,2,\ldots\}$ we have
\begin{align}
\label{rhomon}
\upx [\rho^{2\ell}M_{k}(x)] & = 2\ell\rho^{2\ell-2}\underline{x}M_{k}(x)\\
\label{rhoxmon}
\upx [\rho^{2\ell}\underline{x}M_{k}(x)] & = -2(\ell+k+\frac{m}{2})\rho^{2\ell}M_{k}(x)
\end{align}
with $M_{k}(x) \in \cM_k(\mR^m, \cC \ell_{0,m})$.
\end{corollary}


\section{Null-solutions of the parabolic Dirac operator}
\label{Sec3}

Now we will investigate the null-solutions of the parabolic Dirac operator $D_{\underline{x},t}$, i.e. functions $F(x,t)$ satisfying
\[
D_{\underline{x},t}F=0.
\] 
We will do this by means of series expansions of the component functions $F^{[i]}(x,t)$, $i=0, \ldots, 3$. By separating the radial variable $\rho = |\ux|$, the general terms will be of the form
\[
\rho^{2\ell}M_{k}(x)a_{\ell}(t)+\rho^{2\ell}\underline{x}M_{k}(x)b_{\ell}(t),
\] 
where $M_{k}(x) \in \cM_k(\mR^m, \cC \ell_{0,m})$ is a spherical monogenic of degree $k$ and $a_{\ell}(t)$, $b_{\ell}(t)$, $\ell\in\mathbb{N}$, are $\mathcal{C}\ell_{0,m}$-valued functions.

From Corollary \ref{cortechn} we observe the following relations:
\begin{equation}
\label{relations}
\left\{
            \begin{array}{lcl}
\upx (\rho^{2\ell}M_{k}(x)a(t))&=&2\ell\rho^{2\ell-2}\underline{x}M_{k}(x)a(t)\\
\upx (\rho^{2\ell}\underline{x}M_{k}(x)b(t))&=&-2(\ell+k+\frac{m}{2})\rho^{2\ell}M_{k}(x)b(t)\\
\partial_{t}(\rho^{2\ell}M_{k}(x)a(t))&=&\rho^{2\ell}M_{k}(x)a'(t)\\
\partial_{t}(\rho^{2\ell}\underline{x}M_{k}(x)b(t))&=&\rho^{2\ell}\underline{x}M_{k}(x)b'(t).
\end{array}
\right.
\end{equation}
We combine (\ref{relations}) with the necessary and sufficient condition of Theorem \ref{conditionsPDE}. This will give us conditions under which a $\mathcal{C}\ell_{1,m+1}$-valued function of the form
\begin{align}
F(x,t)&=\sum\limits^{\infty}_{\ell=0}
\rho^{2\ell}\left[
(M_{k}(x)a_{\ell}^{0}(t)+\underline{x}M_{k}(x)b^{0}_{\ell}(t))+
\f(M_{k}(x)a_{\ell}^{1}(t)+\underline{x}M_{k}(x)b^{1}_{\ell}(t))\right.\label{2.5}\\
&\phantom{=\sum\limits^{\infty}_{\ell=0}\rho^{2\ell}[(M_{k}}+\left.
\fd(M_{k}(x)a_{\ell}^{2}(t)+\underline{x}M_{k}(x)b^{2}_{\ell}(t))
+\f\fd(M_{k}(x)a_{\ell}^{3}(t)+\underline{x}M_{k}(x)b^{3}_{\ell}(t))
\right]\nonumber
\end{align}
is a null-solution of the parabolic Dirac operator. At this point we assume $F\in \mathcal{C}^{[2,\infty]}$ and absolute convergence of the series in (\ref{2.5}).

The function
\[F^{[0]}(x,t)=\sum\limits_{\ell=0}^{\infty}\rho^{2\ell}(M_{k}(x)a_{\ell}^{0}(t)+\underline{x}M_{k}(x)b_{\ell}^{0}(t))
\]
satisfies the heat equation if and only if
\begin{align*}
(\Delta-\partial_{t})F^{[0]}=&-\upx^{2}F^{[0]}-\partial_{t}F^{[0]}\\
=&-\upx \left[\sum\limits^{\infty}_{\ell=1}2\ell\rho^{2(\ell-1)}\underline{x}M_{k}(x)a^{0}_{\ell}(t)
-\sum\limits^{\infty}_{\ell=0}2\left(\ell+k+\frac{m}{2}\right)\rho^{2\ell}M_{k}(x)b^{0}_{\ell}(t)\right]
\\
&-\sum\limits^{\infty}_{\ell=0}\rho^{2\ell}\left[M_{k}(x)(a^{0}_{\ell})'(t)+\underline{x}M_{k}(x)(b^{0}_{\ell})'(t)\right]\\
=&\sum\limits_{\ell=0}^{\infty}\rho^{2\ell}
\left[4(\ell+1)\left(\ell+k+\frac{m}{2}\right)M_{k}(x)a_{\ell+1}^{0}(t)
+4(\ell+1)\left(\ell+k+\frac{m+2}{2}\right)\underline{x}M_{k}(x)b_{\ell+1}^{0}(t)\right.\\
&\phantom{\sum\limits_{\ell=0}^{\infty}} \left.\phantom{\left(\frac{m}{2}\right)} -M_{k}(x)(a_{\ell}^{0})'(t)
-\underline{x}M_{k}(x)(b_{\ell}^{0})'(t)\right].
\end{align*}
Therefore, the functions $a^{0}_{\ell}(t)$ and $b^{0}_{\ell}(t)$ have to fulfil the recurrence relations
\begin{equation*}
\left\{
\begin{array}{ll}
(a_{\ell}^{0})'(t)=&4(\ell+1)(\ell+k+\frac{m}{2})a_{\ell+1}^{0}(t)\\
(b_{\ell}^{0})'(t)=&4(\ell+1)(\ell+k+\frac{m+2}{2})b_{\ell+1}^{0}(t)
\end{array}
\right.
\end{equation*}
and thus
\begin{equation*}
  \left\{
\begin{array}{ll}
a^{0}_{\ell+1}(t)&=\frac{(a^{0}_{\ell})'(t)}{4(\ell+1)(\ell+k+\frac{m}{2})}\\
              &=\frac{(a^{0}_{\ell-1})^{''}(t)}{4(\ell+1)(\ell+k+\frac{m}{2})4\ell(\ell-1+k+\frac{m}{2})}\\
             & =\cdots\\
              &=\frac{(a^{0}_{0})^{(\ell+1)}(t)}{4^{(\ell+1)}(\ell+1)!(\ell+k+\frac{m}{2})\cdots(0+k+\frac{m}{2})}\\
(b_{\ell+1}^{0})(t)&=\frac{(b^{0}_{0})^{(\ell+1)}(t)}{4^{(\ell+1)}(\ell+1)!(\ell+k+\frac{m+2}{2})\cdots(0+k+\frac{m+2}{2})}
\end{array}
\right.\notag
\end{equation*}
for $\ell\in\mathbb{N}$. Hence, we have
\begin{equation}\label{a0coeff}
\left\{
\begin{array}{ll}
a_{\ell}^{0}(t)=&\frac{(a^{0}_{0})^{(\ell)}(t)}{4^{\ell}\ell!(k+\frac{m}{2})_{\ell}}\\
b_{\ell}^{0}(t)=&\frac{(b^{0}_{0})^{(\ell)}(t)}{4^{\ell}\ell!(k+\frac{m+2}{2})_{\ell}}
\end{array}
\right.
\end{equation}
for  $\ell\in\mathbb{N}$. A similar reasoning shows
\begin{equation}
\label{a2coeff}
\left\{
\begin{array}{ll}
a_{\ell}^{2}(t)=&\frac{(a^{2}_{0})^{(\ell)}(t)}{4^{\ell}\ell!(k+\frac{m}{2})_{\ell}}\\
b_{\ell}^{2}(t)=&\frac{(b^{2}_{0})^{(\ell)}(t)}{4^{\ell}\ell!(k+\frac{m+2}{2})_{\ell}}
\end{array}
\right.
\end{equation}
for  $\ell\in\mathbb{N}$. In (\ref{a0coeff}) and (\ref{a2coeff}) we have used the Pochhammer symbol: $(a)_{0}=1$ and $(a)_{k}=a(a+1)\cdots(a+k-1)$  for $k\in\mathbb{N}\setminus\{0\}$. Summarising, $F^{[0]}$ and $F^{[2]}$ can be written as
\begin{align}
F^{[0]}(x,t)=&\sum\limits_{\ell=0}^{\infty}\rho^{2\ell}(M_{k}(x)a_{\ell}^{0}(t)+\underline{x}M_{k}(x)b_{\ell}^{0}(t))\notag\\
=&\sum\limits_{\ell=0}^{\infty}\frac{\rho^{2\ell}}{4^{\ell}\ell!(k+\frac{m}{2})_{\ell}}M_{k}(x)(a^{0}_{0})^{(\ell)}(t)
+\underline{x}\sum\limits_{\ell=0}^{\infty}\frac{\rho^{2\ell}}{4^{\ell}\ell!(k+\frac{m+2}{2})_{\ell}}M_{k}(x)(b^{0}_{0})^{(\ell)}(t)\notag\\
=&\phantom{|}_{0}F_{1}\left(k+\frac{m}{2};\frac{\rho^{2}s}{4}\right)M_{k}(x)a^{0}_{0}(t)+ {}_{0}F_{1}\left(k+\frac{m+2}{2};\frac{\rho^{2}s}{4}\right) \underline{x} M_{k}(x)b^{0}_{0}(t)
\label{2.9}\end{align}
and
\begin{align}
F^{[2]}(x,t)&={}_{0}F_{1}\left(k+\frac{m}{2};\frac{\rho^{2}s}{4}\right)M_{k}(x)a^{2}_{0}(t)+ {}_{0}F_{1}\left(k+\frac{m+2}{2};\frac{\rho^{2}s}{4}\right) \underline{x} M_{k}(x)b^{2}_{0}(t)
\label{2.10}
\end{align}
where $s=\partial_t$ and ${}_{0}F_{1}(\gamma;x)=\sum_{n=0}^{\infty}\frac{x^{n}}{n!(\gamma)_{n}}$ is a hypergeometric function.
Using Theorem \ref{conditionsPDE} we can generate the remaining components $F^{[1]}$ and $F^{[3]}$. From $F^{[1]}(x,t)=-\upx F^{[0]}(x,t)$ we find
\begin{align*}
\sum\limits_{\ell=0}^{\infty}\rho^{2\ell}(M_{k}(x)a_{\ell}^{1}(t)+\underline{x}M_{k}(x)b_{\ell}^{1}(t))
=&-\upx \left(\sum\limits_{\ell=0}^{\infty}\rho^{2\ell}(M_{k}(x)a_{\ell}^{0}(t)+\underline{x}M_{k}(x)b_{\ell}^{0}(t)\right)\\
=&\sum\limits_{\ell=0}^{\infty}2\left(\ell+k+\frac{m}{2}\right)\rho^{2\ell}M_{k}(x)b_{\ell}^{0}(t)
-\underline{x}\sum\limits_{\ell=0}^{\infty}2(\ell+1)\rho^{2\ell}M_{k}(x)a_{\ell+1}^{0}(t).
\end{align*}
Hence we obtain
\begin{equation}\label{2.11}
\left\{
\begin{array}{ll}
a^{1}_{\ell}(t)=&2(\ell+k+\frac{m}{2})b^{0}_{\ell}(t)=\frac{2(k+\frac{m}{2})}{4^{\ell}\ell!(k+\frac{m}{2})_{\ell}}(b_{0}^{0})^{(\ell)}(t)\\
b^{1}_{\ell}(t)=&-2(\ell+1)a^{0}_{\ell+1}(t)=-\frac{1}{(2k+m)4^{\ell}\ell!(k+\frac{m+2}{2})_{\ell}}(a_{0}^{0})^{(\ell+1)}(t)
\end{array}
\right.
\end{equation}
for $\ell\in\mathbb{N}$. This leads to
\begin{align}\label{2.12}
F^{[1]}(x,t)=&\sum\limits_{\ell=0}^{\infty}\rho^{2\ell}(M_{k}(x)a_{\ell}^{1}(t)+\underline{x}M_{k}(x)b_{\ell}^{1}(t))\notag\\
            =&(2k+m)_{0}F_{1}\left(k+\frac{m}{2};\frac{\rho^{2}s}{4}\right)M_{k}(x)b^{0}_{0}(t)
            -\frac{\underline{x}}{(2k+m)}{}_{0}F_{1}\left(k+\frac{m+2}{2};\frac{\rho^{2}s}{4}\right)sM_{k}(x)a^{0}_{0}(t).
\end{align}
From $F^{[3]}=\upx F^{[2]}(x,t)-F^{[0]}(x,t)$, we have
\begin{align*}
\sum\limits_{\ell=0}^{\infty}\rho^{2\ell}\left(M_{k}(x)a_{\ell}^{3}(t)+\underline{x}M_{k}(x)b_{\ell}^{3}(t)\right)
=&\phantom{|}\upx \left(\sum\limits_{\ell=0}^{\infty}\rho^{2\ell}\left(M_{k}(x)(a_{\ell}^{2}(t)+\underline{x}M_{k}(x)b_{\ell}^{2}(t)\right)\right)\\
&-\sum\limits_{\ell=0}^{\infty}\rho^{2\ell}\left(M_{k}(x)(a_{\ell}^{0}(t)+\underline{x}M_{k}(x)b_{\ell}^{0}(t)\right)\\
=&\sum\limits_{\ell=0}^{\infty}\rho^{2\ell}\left(-2\left(\ell+k+\frac{m}{2}\right)M_{k}(x)b_{\ell}^{2}(t)-M_{k}(x)a_{\ell}^{0}(t)\right.\\
&\left.\phantom{\left(\frac{m}{2}\right)}+\underline{x}\left[2(\ell+1)M_{k}(x)a_{\ell+1}^{2}(t)-M_{k}(x)b_{\ell}^{0}(t)\right]\right).
\end{align*}
We obtain
\begin{equation}\label{2.13}
\left\{
\begin{array}{ll}
a^{3}_{\ell}(t)=&-2(\ell+k+\frac{m}{2})b^{2}_{\ell}(t)-a^{0}_{\ell}(t)\\
b^{3}_{\ell}(t)=&2(\ell+1)a^{2}_{\ell+1}(t)-b^{0}_{\ell}(t)
\end{array}
\right.
\end{equation}
for $\ell\in\mathbb{N}$. Substituting (\ref{a0coeff}) and (\ref{a2coeff}) into (\ref{2.13}), yields
\begin{equation}\label{2.14}
\left\{
\begin{array}{ll}
a^{3}_{\ell}(t)&=-\frac{1}{4^{\ell}\ell!(k+\frac{m}{2})_{\ell}}[(2k+m)(b_{0}^{2})^{(\ell)}(t)+(a_{0}^{0})^{(\ell)}(t)]\\
b^{3}_{\ell}(t)&=\frac{1}{4^{\ell}\ell!(k+\frac{m+2}{2})_{\ell}}[\frac{(a_{0}^{2})^{(\ell+1)}(t)}{2k+m}-(b_{0}^{0})^{(\ell)}(t)]
\end{array}
\right.
\end{equation}
This leads to
\begin{align}\label{2.15}
F^{[3]}(x,t)=&\phantom{|}_{0}F_{1}\left(k+\frac{m}{2};\frac{\rho^{2}s}{4}\right)\left(-(2k+m)M_{k}(x)b^{2}_{0}(t)-M_{k}(x)a^{0}_{0}(t)\right)\notag\\
&+ {}_{0}F_{1}\left(k+\frac{m+2}{2};\frac{\rho^{2}s}{4}\right)\underline{x}
\left(\frac{s}{2k+m}M_{k}(x)a^{2}_{0}(t)-M_{k}(x)b^{0}_{0}(t)\right).
\end{align}
Substituting (\ref{2.9}), (\ref{2.10}), (\ref{2.12}) and (\ref{2.15}) into (\ref{2.1}), we get the following null-solutions of the parabolic Dirac operator $D_{\ux,t}$ by means of series expansions with general term $\rho^{2\ell}[M_{k}(x)a_{\ell}(t)+\underline{x}M_{k}(x)b_{\ell}(t)]$:
\begin{align}\label{2.16}
F(x,t)=&F^{[0]}(x,t)+\f F^{[1]}(x,t)+\fd F^{[2]}(x,t)+\f\fd F^{[3]}(x,t)\\
=&{}_{0}F_{1}\left(k+\frac{m}{2};\frac{\rho^{2}s}{4}\right)\left[(1-\f \fd)M_{k}(x)a^{0}_{0}(t)
+\fd M_{k}(x)a^{2}_{0}(t)\right.\notag\\
&\phantom{{}_{0}F_{1}\left(k+\frac{m}{2};\frac{\rho^{2}s}{4}\right)qdf}\left.+(2k+m)(\f M_{k}(x)b^{0}_{0}(t)-\f\fd M_{k}(x)b^{2}_{0}(t))\right]\notag\\
&+ {}_{0}F_{1}\left(k+\frac{m+2}{2};\frac{\rho^{2}s}{4}\right) \ux \left[(1-\f\fd)M_{k}(x)b^{0}_{0}(t)-\fd M_{k}(x)b^{2}_{0}(t)\phantom{\frac{s}{(2k+m)}}\right.\notag\\
&\phantom{{}_{0}F_{1}\left(k+\frac{m+2}{2};\frac{\rho^{2}s}{4}\right) \ux qdqds}\left.+\frac{s}{(2k+m)}(\f M_{k}(x)a^{0}_{0}(t)+\f\fd M_{k}(x)a^{2}_{0}(t))\right]\notag.
\end{align}

It is now possible to formulate the following theorem, which presents the solution (\ref{2.16}) in a much more concise way.
\begin{theorem}
\label{thmPara}
The general solution of the parabolic Dirac equation $D_{\underline{x},t} F(x,t)=0$ of type (\ref{2.5}) is given by
\begin{align}\label{2.17}
F(x,t)=\left[1-\left(\f+\frac{1}{s}\fd\right)\upx \right]{}_0F_1\left(k+\frac{m}2;\frac{\rho^2s}{4}\right)M_{k}(x)a(t).
\end{align}
where $s = \partial_t$, $M_{k}(x) \in \cM_k(\mR^m, \cC \ell_{0,m})$ is a spherical monogenic of degree $k$ and $a(t)$ is a $\mathcal{C}\ell_{1,m+1}$-valued function.
\end{theorem}

\begin{proof}
By taking $a^0_0(t)=a(t)$, $b^2_0=-a(t)/(2k+m)$ and $a^2_0=b^0_0=0$ in (\ref{2.16}), the null-solution for the parabolic Dirac operator can be expressed as
\begin{align}\label{2.18}
F(x,t)=&\phantom{|}_{0}F_{1}\left(k+\frac{m}{2};\frac{\rho^{2}s}{4}\right)\left[(1-\f\fd)M_{k}(x)a(t)+\f\fd M_{k}(x)a(t)\right]\notag\\
   &+{}_{0}F_{1}\left(k+\frac{m+2}{2};\frac{\rho^{2}s}{4}\right) \ux \left[\frac{1}{2k+m}\fd M_{k}(x)a(t)+\frac{s}{2k+m}\f M_{k}(x)a(t)\right]\notag\\
   =&\phantom{|}_{0}F_{1}\left(k+\frac{m}{2};\frac{\rho^{2}s}{4}\right)M_{k}(x)a(t)+{}_{0}F_{1}\left(k+\frac{m+2}{2};\frac{\rho^{2}s}{4}\right) \ux \left[\frac{\fd +s \f}{2k+m}M_{k}(x)a(t)\right]\notag\\
   =&\left[{}_{0}F_{1}\left(k+\frac{m}{2};\frac{\rho^{2}s}{4}\right)+{}_{0}F_{1}\left(k+\frac{m+2}{2};\frac{\rho^{2}s}{4}\right) \ux \frac{\fd +s \f}{2k+m}\right]M_{k}(x)a(t).
\end{align}
Since
\[
\upx \left({}_0F_1\left(k+\frac{m}2;\frac{\rho^2s}{4}\right)M_{k}(x)a(t)\right)=\frac{\underline{x}s}{k+2m}{}_0F_1\left(k+\frac{m+2}2;\frac{\rho^2s}{4}\right)M_{k}(x)a(t),
\]
we have
\begin{align*}
F(x,t)=&\phantom{|}_{0}F_{1}\left(k+\frac{m}{2};\frac{\rho^{2}s}{4}\right)M_{k}(x)a(t)-\f\upx \left( {}_{0}F_{1}\left(k+\frac{m}{2};\frac{\rho^{2}s}{4}\right)M_{k}(x)a(t)\right)\\
&-\frac{\fd }{s}\upx \left( {}_{0}F_{1}(k+\frac{m}{2};\frac{\rho^{2}s}{4})M_{k}(x)a(t)\right)\\
=&\left[1-\left(\f+\frac1s\fd\right)\upx \right]\left({}_0F_1\left(k+\frac{m}2;\frac{\rho^2s}{4}\right)M_{k}(x)a(t)\right).
\end{align*}
Hence for suitable choices of the input functions, (\ref{2.16}) indeed reduces to (\ref{2.17}).

Now we show that any solution of the form (\ref{2.17}) is of the type (\ref{2.16}). If $a(t)$ is not restricted to $\cC\ell_{0,m}$ but allowed to vary in $\cC\ell_{1,m+1}$, we verify that
\begin{align*}
 & \left({}_{0}F_{1}\left(k+\frac{m}{2};\frac{\rho^{2}s}{4}\right)+\frac{\underline {x}}{2k+m}(s\f+\fd ){}_{0}F_{1}\left(k+\frac{m+2}{2};\frac{\rho^{2}s}{4}\right)\right)M_{k}(x)(1-\f\fd )\\
 & =\left({}_{0}F_{1}\left(k+\frac{m}{2};\frac{\rho^{2}s}{4}\right)(1-\f\fd )+\frac{\underline{x}sf}{2k+m}
  {}_{0}F_{1}\left(k+\frac{m+2}{2};\frac{\rho^{2}s}{4}\right)\right)M_{k}(x),
\end{align*}
which is the coefficient of $a^{0}_{0}(t)$ in the general solution (\ref{2.16}). Similarly, the coefficient of $a^{2}_{0}(t)$ is obtained by multiplying (\ref{2.18}) to the right by $\fd $. Analogously, multiplying to the right by $(m+2k)\f$ and  $-\f\fd $ gives the coefficient of $b^{0}_{0}(t)$ and $b^{2}_{0}(t)$ respectively.
\end{proof}


\section{Solutions of the generalised parabolic Dirac operator}\label{Sec4}

Definition \ref{def-com} inspires us to define a more general class of Dirac operators as follows. Introduce the following element of $\cC \ell_{1,1}$:
\[
\zeta=a\f \fd+ b\f+c \fd+d \fd \f
\]
where $a, b, c, d$ are complex constants. 

\begin{definition}
\label{def-com2}
We define the generalised parabolic Dirac operator as
$$ \gpde$$
where $\zeta=a\f \fd+ b\f+c \fd+d \fd \f$ with $a, b, c, d\in \mC$.
\end{definition}

As far as we are aware, the operator of Definition \ref{def-com2} has not been studied before.

We state some direct observations. Let $\zeta$ be as above, then it can be
represented by $\left(\begin{array}{cc} a&b\\c&d\end{array}\right)$. Its main involution is given by
\[
\zeta^{\ast} = a\f\fd - b\f - c\fd + d\fd\f
\]
This is represented by the matrix
\[
\left(\begin{matrix}
a & -b\\
-c & d
\end{matrix}\right) = \left(\begin{matrix}
1 & 0\\
0 & -1
\end{matrix}\right) \left(\begin{matrix}
a & b\\
c & d
\end{matrix}\right) \left(\begin{matrix}
1 & 0\\
0 & -1
\end{matrix}\right)  
\]
or in terms of the Clifford elements:
\[
\zeta^{\ast} = (\f\fd - \fd\f)\zeta(\f\fd - \fd\f)
\]
so that $\zeta^{\ast}\zeta=\xi^2$ where

\[
\xi =(\f\fd-\fd\f)\zeta=a\f\fd + b\f - c\fd - d\fd\f.
\]
Thus $\xi$ is represented by $A=\left(\begin{array}{cc}
    a & b\\-c & -d\end{array}\right)$.
       
We are looking for solutions of the generalised Dirac equations, i.e. functions $g$ that satisfy 
\begin{equation}\label{genDireq}
(\upx+\zeta)g=0.
\end{equation}
Consequently, they also satisfy $(-\upx+\zeta^{\ast})(\upx+\zeta)g=0$, hence they are solutions of the generalised Helmholtz equation
\begin{equation}\label{genHelmeq}
(\Delta+\zeta^{\ast}\zeta)g=0.
\end{equation}
By separating the radial variable $\rho=|\ux|$, we search for solutions of the form $h_k(\rho)H_k(x)$ where $H_k\in\mathcal{H}_k(\mR^m,\cC \ell_{0,m})$ and $h_k(\rho)$ is an entire function.

\begin{proposition}
The general solution of the generalised Helmholtz equation (\ref{genHelmeq}) as a series of harmonic polynomials is given by

\[
g(x)=\sum_{k=0}^{+\infty} {}_0F_1\left(\frac{m}{2}+k;-\frac{1}{4}\zeta^{\ast}\zeta\rho^2\right)  H_k(x).
\]
\end{proposition}
\begin{proof}
 By change of variables we find for a radial, $\mathcal{C}_{1,1}$-valued, entire function $f(\rho)$:

\[
\partial_{\ul{x}} f(\rho) = \frac{1}{\rho}\ux f'(\rho) = \left(f'(\rho)\right)^{\ast}\frac{1}{\rho}\ux.
\]
If we now use $\ux\partial_{\ul{x}} + \partial_{\ul{x}}\ux = -2\mathbb{E} - m$, where $\mathbb{E} = \sum_{i=1}^m x_i \partial_{x_i}$ is the Euler operator, we get

\begin{align*}
\Delta \left[h_k(\rho)H_k(x)\right] =& -\partial_{\ul{x}}^2 \left[h_k(\rho)H_k(x)\right]\\
=& -\partial_{\ul{x}} \left[(\partial_{\ul{x}} h_k(\rho)) H_k(x) + (h_k(\rho))^{\ast} (\partial_{\ul{x}} H_k(x))\right]\\
=& -\partial_{\ul{x}} \left[\frac{(h_k'(\rho)^{\ast}}{\rho} \ux H_k(x) + (h_k(\rho))^{\ast} (\partial_{\ul{x}} H_k(x))\right]\\
=& \left(\frac{h_k''(\rho)}{\rho} - \frac{h_k'(\rho)}{\rho^2}\right)\left(-\frac{\ux^2}{\rho}\right)H_k(x) + \frac{h_k'(\rho)}{\rho} (2\mathbb{E} + m)H_k\\
=& \left(h_k''(\rho) - \frac{h_k'(\rho)}{\rho}\right)H_k(x) + (2k+m)\frac{h_k'(\rho)}{\rho} H_k.
\end{align*}
Hence $(\Delta+\zeta^{\ast}\zeta)\left[h_k(\rho)H_k(x)\right]=0$ if $h_k(\rho)$ satisfies $\displaystyle \frac{d^2h_k}{d\rho^2}+\frac{m+2k-1}{\rho}\frac{dh_k}{d\rho}+\zeta^{\ast}\zeta h_k=0$. Writing $h_k(\rho) = \sum_{n=0}^{\infty} \theta_n \rho^n$, we obtain:

\begin{align*}
0 &= \frac{d^2h_k}{d\rho^2}+\frac{m+2k-1}{\rho}\frac{dh_k}{d\rho}+\zeta^{\ast}\zeta h_k\\
&= \sum_{n=2}^{\infty} \theta_n n (n-1) \rho^{n-2} + \frac{m+2k-1}{\rho} \sum_{n=1}^{\infty} \theta_n n \rho^{n-1} + \zeta^{\ast}\zeta \sum_{n=0}^{\infty} \theta_n \rho^n\\
&= \sum_{n=0}^{\infty} (\theta_{n+2}(n+2)(n+m+2k) + \zeta^{\ast} \zeta \theta_n) \rho^n + \theta_1 (m+2k-1) \frac{1}{\rho}.
\end{align*}
This implies

\[
\left\{\begin{array}{l}
\theta_{n+2}(n+2)(n+m+2k) = - \zeta^{\ast} \zeta \theta_n\\
\theta_1 = 0
\end{array}\right.
\]
whence  $h_k(\rho)={}_0F_1\left(\frac{m}{2}+k;-\frac{1}{4}\zeta^{\ast}\zeta\rho^2\right)$. Hence a general solution in terms of harmonic polynomials is given by

\begin{equation}\label{monogenic building blocks genpardirac}
g(x)=\sum_{k=0}^{+\infty} {}_0F_1\left(\frac{m}{2}+k;-\frac{1}{4}\zeta^{\ast}\zeta\rho^2\right)  H_k(x).
\end{equation}

\end{proof}

It is a well-known fact, see \cite[Theorem 4.4, p138]{Perspective}, that each harmonic polynomial has a monogenic decomposition $H_k(x) = M_k(x) + \ux \widetilde{M}_{k-1}(x)$, with $M_k(x)\in\cM_k(\mR^m, \cC \ell_{0,m})$ and $\widetilde{M}_{k-1}(x) \in \cM_{k-1}(\mR^m, \cC \ell_{0,m})$. Thus in terms of monogenic polynomials, the building blocks of solutions to (\ref{genHelmeq}) are ${}_0F_1\left(\frac{m}{2}+k;-\frac{1}{4}\zeta^{\ast}\zeta\rho^2\right)M_k(x)$ and ${}_0F_1\left(\frac{m}{2}+k+1;-\frac{1}{4}\zeta^{\ast}\zeta\rho^2\right)\ux\widetilde{M}_k(x)$, where the $M_k$ and $\widetilde{M}_k$ are monogenic and homogeneous of degree $k$. We are left with considering linear combinations of the form

\[
g(x) = \sum_{k=0}^{+\infty}\left[{}_0F_1\left(\frac{m}{2}+k;-\frac{1}{4}\zeta^{\ast}\zeta\rho^2\right)M_k(x) + {}_0F_1\left(\frac{m}{2}+k+1;-\frac{1}{4}\zeta^{\ast}\zeta\rho^2\right)x\widetilde{M}_k(x)\right].
\]
as potential solutions for the generalised parabolic Dirac equation.

\begin{proposition}\label{Solution generalized Dirac mono}
A general solution of the generalised parabolic Dirac equation (\ref{genDireq}) expressed as a series of monogenic polynomials is given by

\[
g(x) = \sum_{k=0}^{+\infty}\left[{}_0F_1\left(\frac{m}{2}+k;-\frac{1}{4}\zeta^{\ast}\zeta\rho^2\right) + \frac{{}_0F_1\left(\frac{m}{2}+k+1;-\frac{1}{4}\zeta^{\ast}\zeta\rho^2\right)\ux\zeta}{m+2k}\right]M_k(x).
\] 

\end{proposition}
\begin{proof}
 Applying the Dirac operator to the building blocks (\ref{monogenic building blocks genpardirac}) in terms of monogenic polynomials, yields

\begin{align}
\upx \left({}_0F_1\left(\frac{m}{2}+k;-\frac{1}{4}\zeta^{\ast}\zeta\rho^2\right)M_k(x)\right) &= -2\sum_{n=1}^{\infty} \left(\frac{\zeta\zeta^{\ast}}{4}\right)^n \frac{\ux^{2n-1}}{(k+\frac{m}{2})_n (n-1)!} M_k(x)\nonumber\\
&=\frac{\zeta\zeta^{\ast}}{m+2k}{}_0F_1\left(\frac{m}{2}+k+1;-\frac{1}{4}\zeta\zeta^{\ast} \rho^2\right)\ux M_k(x) \label{Dirac operator of 0F1(m/2+k)}\\
\upx \left({}_0F_1\left(\frac{m}{2}+k+1;-\frac{1}{4}\zeta^{\ast}\zeta\rho^2\right)\ux\widetilde{M}_k(x)\right) &= \sum_{n=0}^{\infty} \left(\frac{\zeta\zeta^{\ast}}{4}\right)^n \frac{(-2n-2k-m)\ux^{2n}}{(k+1+\frac{m}{2})_n n!} \widetilde{M}_k(x) \nonumber
\end{align}
Requiring that $(\partial_{\ul{x}}+\zeta)g=0$ yields the following equation for each $k$:

\begin{align*}
0=&(\gpde)\left[{}_0F_1\left(\frac{m}{2}+k;-\frac{1}{4}\zeta^{\ast}\zeta\rho^2\right)M_k(x) + {}_0F_1\left(\frac{m}{2}+k+1;-\frac14\zeta^{\ast}\zeta\rho^2\right)\ux\widetilde{M}_k(x)\right] \\
=& \sum_{n=1}^{\infty} \left(\frac{\zeta\zeta^{\ast}}{4}\right)^n \frac{(-2n)\ux^{2n-1}}{(k+\frac{m}{2})_n n!}M_k(x) + \sum_{n=0}^{\infty} \left(\frac{\zeta\zeta^{\ast}}{4}\right)^n \frac{(-2n-2k-m)\ux^{2n}}{(k+1+\frac{m}{2})_n n!}\widetilde{M}_k(x)\\
&+\sum_{n=0}^{\infty} \zeta\left(\frac{\zeta^{\ast}\zeta}{4}\right)^n \frac{\ux^{2n}}{(k+\frac{m}{2})_n n!}M_k(x) + \sum_{n=0}^{\infty} \zeta\left(\frac{\zeta^{\ast}\zeta}{4}\right)^n \frac{\ux^{2n+1}}{(k+1+\frac{m}{2})_n n!}\widetilde{M}_k(x)\\
=& \sum_{n=0}^{\infty} \left(\frac{\zeta\zeta^{\ast}}{4}\right)^{n+1} \frac{(-2)\ux^{2n+1}}{(k+\frac{m}{2})_{n+1} n!}M_k + \sum_{n=0}^{\infty} \left(\frac{\zeta\zeta^{\ast}}{4}\right)^n \frac{(-2n-2k-m)\ux^{2n}}{(k+1+\frac{m}{2})_n n!}\widetilde{M}_k(x)\\
&+\sum_{n=0}^{\infty} \left(\frac{\zeta\zeta^{\ast}}{4}\right)^n \frac{\ux^{2n}}{(k+\frac{m}{2})_n n!}\zeta M_k + \sum_{n=0}^{\infty} \left(\frac{\zeta\zeta^{\ast}}{4}\right)^n \frac{\ux^{2n+1}}{(k+1+\frac{m}{2})_n n!}\zeta^{\ast}\widetilde{M}_k(x).
\end{align*}
By comparing the factors of each power of $\ux$ we get

\[
\left\{ \begin{array}{l}
\dfrac{\zeta}{(k+\frac{m}{2})_n} M_k(x) = \dfrac{2(n+k+\frac{m}{2})}{(k+1+\frac{m}{2})_n}\widetilde{M}_k(x)\\[0.4cm]
\dfrac{\zeta^{\ast}}{(k+1+\frac{m}{2})_n} \widetilde{M}_k(x) = \dfrac{\zeta^{\ast}\zeta}{2(k+\frac{m}{2})_{n+1}}M_k(x)
\end{array}\right.
\]
which yields

\[
\widetilde{M}_k(x) = \dfrac{\zeta}{2k+m} M_k(x).
\]
This allows to finally express the solutions as 
         
\[
g(x) = \sum_{k=0}^{+\infty}\left[{}_0F_1\left(\frac{m}{2}+k;-\frac{1}{4}\zeta^{\ast}\zeta\rho^2\right) + \frac{{}_0F_1\left(\frac{m}{2}+k+1;-\frac14\zeta^{\ast}\zeta\rho^2\right)\ux\zeta}{m+2k}\right]M_k(x).
\] 

\end{proof}

Rewriting Proposition \ref{Solution generalized Dirac mono} yields the following theorem:

\begin{theorem}\label{thmGenPara}
The solution of the generalised parabolic Dirac equation (\ref{genDireq}) as a series of spherical monogenics is given by

\[
g(x)=\sum_{k=0}^{+\infty}(\zeta^{\ast}-\upx)\left[{\displaystyle \dfrac{\left({}_0F_1\left(\dfrac{m}{2}+k+1,-\dfrac{1}{4}\zeta^{\ast}\zeta\rho^2\right)\right)^{\ast}}{m+2k}}\ux M_k(x)\right]
\]
\end{theorem}
\begin{proof}
Let $g(x)$ be of the form found in Proposition \ref{Solution generalized Dirac mono}, i.e.
\[
g(x) = \sum_{k=0}^{+\infty}\left[\underbrace{{}_0F_1\left(\frac{m}{2}+k;-\frac{1}{4}\zeta^{\ast}\zeta\rho^2\right)}_{:=A} + \underbrace{\frac{{}_0F_1\left(\frac{m}{2}+k+1;-\frac{1}{4}\zeta^{\ast}\zeta\rho^2\right)\ux\zeta}{m+2k}}_{:=B}\right]M_k(x).
\] 
Note that

\[
\partial_{\ul{x}} \left(\frac{{}_0F_1\left(\frac{m}{2}+k+1;-\frac{1}{4}\zeta\zeta^{\ast}\rho^2\right)}{m+2k} \ux M_k(x)\right) = -{}_0F_1\left(\frac{m}{2}+k;-\frac{1}{4}\zeta^{\ast}\zeta\rho^2\right)M_k(x).
\]
Hence $A = -\upx\left(\frac{{}_0F_1\left(\frac{m}{2}+k+1;-\frac{1}{4}\zeta\zeta^{\ast}\rho^2\right)}{m+2k} x M_k(x)\right)$. The term $B$ on the other hand, can be rewritten as

\[
B = \zeta^{\ast}\frac{{}_0F_1\left(\frac{m}{2}+k+1;-\frac{1}{4}\zeta\zeta^{\ast}\rho^2\right)}{m+2k} \ux.
\]
Combining these results yields

\[
g(x)=\sum_{k=0}^{+\infty}(\zeta^{\ast}-\partial_x)\left[{\displaystyle \dfrac{\left({}_0F_1\left(\dfrac{m}{2}+k+1;-\dfrac{1}{4}\zeta^{\ast}\zeta\rho^2\right)\right)^{\ast}}{m+2k}}\ux M_k(x)\right].
\]
\end{proof}

For invertible $\zeta$, we can rework Proposition \ref{Solution generalized Dirac mono} as follows.

\begin{theorem}\label{thmGenPara zeta invertible}
If $\zeta$ is invertible, then

\[
g(x)=\sum_{k=0}^{+\infty}(1-\zeta^{-1}\upx)\left[{}_0F_1\left(\frac{m}{2}+k;-\frac{1}{4}\zeta^{\ast}\zeta\rho^2\right)M_k(x)\right]
\]
\end{theorem}
\begin{proof}
Recall the result obtained in (\ref{Dirac operator of 0F1(m/2+k)})
\[
\upx \left({}_0F_1\left(\frac{m}{2}+k;-\frac{1}{4}\zeta^{\ast}\zeta\rho^2\right)M_k(x)\right) =\frac{\zeta\zeta^{\ast}}{m+2k}\phantom{|}_0F_1\left(\frac{m}{2}+k+1;-\frac{1}{4}\zeta\zeta^{\ast} \rho^2\right)\ux M_k(x)
\]
Thus if $\zeta$ is invertible we have

\begin{align*}
\zeta^{-1}\upx \left({}_0F_1\left(\frac{m}{2}+k;-\frac{1}{4}\zeta^{\ast}\zeta\rho^2\right)M_k(x)\right) &= \frac{\zeta^{\ast}}{m+2k}\phantom{|}_0F_1\left(\frac{m}{2}+k+1;-\frac{1}{4}\zeta\zeta^{\ast} \rho^2\right)\ux M_k(x)\\
&=\frac{{}_0F_1\left(\frac{m}{2}+k+1;-\frac14\zeta^{\ast}\zeta\rho^2\right)\ux\zeta}{m+2k} M_k(x)
\end{align*}
which proves the result.
\end{proof}

At the beginning of this section, we introduced a matrix representation for both $\zeta$ and $\xi$. In doing so, we can interpret the previous two results using Sylvester's formula (see \cite{Sylvester}):
    
\begin{lemma}	
Let $f$ be an analytic function and let $B$ be a diagonalisable matrix with distinct eigenvalues $\lambda_i$, $i = 1,\ldots,k$. Then

\[
f(B) = \sum_{i=1}^k f(\lambda_i) B_i
\]
where

\[
B_i = \prod_{\substack{j=1\\ j\neq i}}^k \frac{1}{\lambda_i - \lambda_j} (B-\lambda_j I)
\]
with $I$ the identity matrix.
\end{lemma}
    
\begin{proposition}\label{prop evaluating in eigenvalues}
Let $\lambda_\pm=\frac{1}{2}\left(a-d\pm\sqrt{(a+d)^2-4bc}\right)$ be the eigenvalues of the matrix $A$ representing $\xi$ and let $\psi$ be an entire function, then
    
\[
\psi(\zeta^{\ast}\zeta) = \left\{ \begin{array}{ll}
      					  \displaystyle \frac{\xi-\lambda_-}{\lambda_+-\lambda_-}\psi(\lambda_+^2) + \frac{\xi-\lambda_+}{\lambda_--\lambda_+}\psi(\lambda_-^2), & \text{if }\lambda_+\neq\lambda_-,\\
      					  ~\\
      					  \psi(\lambda^2) + 2\lambda\psi'(\lambda^2)(\xi-\lambda), & \text{if }\lambda_+ = \lambda_- = \lambda. \end{array}
    \right.
\]
\end{proposition}
\begin{proof}
If $\lambda_+\neq\lambda_-$, then we can use Sylvester's formula to find

\[
\psi(A^2) =\frac{1}{\lambda_+ - \lambda_-} (A - \lambda_- I_2)\psi(\lambda_{+}^2) + \frac{1}{\lambda_- - \lambda_+} (A - \lambda_+ I_2)\psi(\lambda_{-}^2).
\]

If $\lambda_+ = \lambda_- = \lambda$, we have two cases:

\begin{enumerate}
\item[(i)]
$A$ is diagonalisable. In this case we have only one possibility for $A$, namely $\lambda I$, i.e. $\xi = \lambda$. Using Sylvester's formula we find

\[
\psi(A^2) = \psi(\lambda^2) I.
\]

\item[(ii)]
$A$ is not diagonalisable. In this case, there exists a matrix $Q$ such that

\[
A = Q \left( \begin{matrix}
\lambda & 1\\
0 & \lambda
\end{matrix}\right) Q^{-1}.
\]
Easy calculations show

\[
A^n = Q \left( \begin{matrix}
\lambda^n & n \lambda^{n-1}\\
0 & \lambda^n
\end{matrix}\right) Q^{-1} = \lambda^nI + Q \left( \begin{matrix}
0 & n \lambda^{n-1}\\
0 & 0
\end{matrix}\right) Q^{-1}.
\]
It is now easy to see that

\[
\psi(A^2) = \psi(\lambda^2) + 2\lambda \psi'(\lambda^2) (A-\lambda I).
\]
\end{enumerate}

\end{proof}

\begin{remark}
When putting $a=d=0$, $c=1$ and $b = \partial_t$ we reobtain the parabolic Dirac operator $D_{\ux,t}$. Here we are slightly abusing notation by allowing $b$ to be a partial derivative with respect to $t$, as this still commutes with $\ux$-variables. If we now apply Proposition \ref{prop evaluating in eigenvalues} with $\lambda_{\pm} = \pm \sqrt{-s}$, where $s = \partial_t$, to the solution of the generalised parabolic Dirac operator, we get the same result as in Section \ref{Sec3}.

\end{remark}

\section{Conclusion and outlook}

In Theorem \ref{thmPara} we have shown that solutions of the parabolic Dirac operator can be written in terms of hypergeometric functions. Using these techniques we proved that solutions of the generalised Dirac operator can be written as a series of hypergeometric functions multiplied with spherical monogenics. If we represent the Clifford algebra $\cC \ell_{1,1}$ by complex-valued $2\times2$ matrices we can use Sylvester's formula to interpret the results.

In future research we will address the question of how the symmetries of \cite{AK1, AK2} act on the solutions of the L\'evy-Leblond equation.

 \section*{Acknowledgements}
This paper was written during a visit of SB to Ghent University, funded by CSC. The work of HDB is supported by the Research Foundation Flanders (FWO) under Grant EOS 30889451.

\end{document}